\documentclass[letterpaper, 10 pt, conference]{ieeeconf}  

\IEEEoverridecommandlockouts                              
\usepackage{amssymb,amsfonts,amsmath,amsthm,amscd,bm,amsthm}
\usepackage{graphicx,rotating,epic}

\newtheorem{fact}{Fact}

\newtheorem{lemma}{Lemma}
\newtheorem{propo}{Proposition}

\def\sTP{{\small \rm TP}}
\def\sBP{{\small \rm BP}}

\def\sMAP{{\small \rm MAP}}

\def\prob{\mathbb P}

\def\ball{{\sf B}}

\def\|{\big|\big|}
\def\sTV{{\tiny \rm TV}}
\def\Tree{{\sf T}}
\def\SAW{{\sf SAW}}
\def\Code{{\mathfrak C}}
\def\ux{\underline{x}}
\def\u0t{{\tt \underline{0}}}
\def\0t{{\tt 0}}
\def\1t{{\tt 1}}
\def\H{{\mathbb H}}
\def\cY{{\cal Y}}
\def\uy{\underline{y}}
\def\cG{{\cal G}}
\def\cV{{\cal V}}
\def\cE{{\cal E}}
\def\ind{{\mathbb I}}
\def\cX{{\cal X}}
\def\reals{{\mathbb R}}
\def\da{{\partial a}}
\def\di{{\partial i}}
\def\uX{\underline{X}}
\def\uY{\underline{Y}}
\def\xh{\hat{x}}
\def\un{\underline{n}}
\def\D{{\sf D}}
\def\tomega{\widehat{\omega}}
\def\cH{{\cal H}}
\def\Rh{\widehat{R}}
\def\SCHEME{{\sf SAW}}
\def\MAPt{{\sf MAP}}
\def\BPt{{\sf BP}}
\def\tl{{\tt l}}
\def\tr{{\tt r}}

\newcommand{\defas}{{\ensuremath{\overset{\text{\tiny def}}{=}}}}

\newcommand{\n}{{\ensuremath{n}}} 

\newcommand{\MAP}{{\ensuremath{\text{\tiny MAP}}}}
\newcommand{\BP}{{\ensuremath{\text{\tiny BP}}}}
\newcommand{\TP}{{\ensuremath{\text{\tiny TP}}}}
\newcommand{\e}{{\ensuremath{\epsilon}}} 

\title{\LARGE \bf TP Decoding}

\author{Yi Lu, Cyril M{\'e}asson and Andrea Montanari
\thanks{Yi Lu is with the Department of Electrical Engineering,
Stanford University, {\tt\small yi.lu@stanford.edu}. Cyril M{\'e}asson is with the Math Center, Bell Labs, Alcatel-Lucent, Murray Hill, {\tt\small cyril.measson@gmail.com}.
Andrea Montanari is with Departments of Electrical Engineering
and Statistics, Stanford University,
{\tt\small montanari@stanford.edu}.}}

\begin{document}

\maketitle
\thispagestyle{empty}
\pagestyle{empty}
\begin{abstract}
`Tree pruning' (TP) is an algorithm for probabilistic inference on
binary Markov random fields. It has been recently derived by Dror
Weitz and used to construct the first fully polynomial
approximation scheme for counting independent sets up to the `tree
uniqueness threshold.' It can be regarded as a clever method for
pruning the belief propagation computation tree, in such a way to
exactly account for the effect of loops.

In this paper we generalize the original algorithm to make it
suitable for decoding linear codes, and discuss various schemes
for pruning the computation tree. Further, we present
the outcomes of numerical simulations on several linear codes,
showing that tree pruning allows to interpolate continuously
between belief propagation and maximum a posteriori decoding.
Finally, we discuss  theoretical implications of the
new method.
\end{abstract}
%
%
\section{Introduction}

Statistical inference is the task of computing marginals
(or expectation values) of complex multi-variate
distributions. \emph{Belief propagation} (BP) is a generic method
for accomplishing this task quickly but approximately,
when the multivariate distribution factorizes according to
a sparse graphical structure.
The advent of sparse graph codes and iterative BP
 decoding \cite{MCT} has naturally made decoding become an
important case of this general problem.
The present paper builds on this connection by `importing' an
algorithm that has been recently developed in the context of
approximate counting and inference \cite{Weitz}.

We will refer to the new algorithm as \emph{tree pruning} (TP)
\emph{decoding}. For a number of reasons the application of this
method to decoding is non-trivial. However, it is an interesting
approach for the three following reasons. $(i)$ It provides a sequence of
decoding schemes that interpolates continuously between BP and the
optimal maximum a posteriori (MAP) decoding. $(ii)$ At each level
of this sequence, the effect of loops of increasing length is
taken into account. $(iii)$ We expect that an appropriate
truncation of this sequence might yield a polynomial algorithm for
MAP decoding on general graphs of bounded degree, for low enough
noise levels. Preliminary numerical results are encouraging.

\subsection{Qualitative Features and Relation to BP Decoding}

As for BP decoding, TP decoding aims at estimating
the a posteriori marginal probabilities of the codeword bits.
Unhappily, the relation between BP estimates and the actual marginals
is in general poorly understood. In the case of random
Low-Density Parity-Check (LDPC) codes and communication over
memoryless channels, density evolution allows to show that,
at small enough noise level, the BP bit error probability
becomes arbitrarily small if the blocklength is
large enough. This implies that the distance between BP estimates
and the actual marginals vanishes as well.
This result is not completely satisfactory, in that
it relies in a crucial way on the locally tree-like
structure of sparse random graphs.
This property does not hold for structured graphs, and, even
for large graphs, it kicks in only at very large blocklengths.

In contrast to this, the algorithm considered in this paper
accounts systematically for short loops. It should therefore
produce better performances, in particular in the error floor
regime since this is dominated by small error events
\cite{Richardson}.

A convenient way of understanding the difference between BP and
MAP decoding makes use of the so-called computation tree. Consider
a code described by a factor graph $G = (V,F,E)$ whereby $V$
represents the variable nodes, $F$ the factor nodes, and $E$ the
edges. Let  $i\in V$, then the corresponding computation tree
denoted by $\Tree(i)$ is the tree of non-reversing walks in $G$
that start at $i$. This gives a graph (tree) structure in a
natural way: two nodes are neighbors if one is reached from the
other adding a step.

BP uses the marginal at the root of $\Tree(i)$ as an estimate for
the marginal distribution at $i$ on the original graph $G$. If $G$
contains short loops in the neighborhood of $i$, the computation
tree differs from $G$ in a neighborhood of the root and, as a
consequence, the BP estimate can differ vastly from the actual
marginal.

Weitz \cite{Weitz} made the surprising remark that there exists a simple
way of pruning the computation tree (and fixing some of its variables)
in such a way that the resulting root marginal coincides
with the marginal on $G$. Unhappily the size of the pruned tree, which we call the
\emph{self-avoiding walk tree} and denote by $\SAW(i)$, is exponential
in the size of the original graph $G$. Nevertheless, the tree can be
truncated thus yielding a convergent sequence of
approximations for the marginal at $i$.
The complexity of the resulting algorithm is linear in the
size of the truncated tree. Its efficiency depends on how sensitive
is the root marginal to the truncation depth.
%
%
\subsection{Contributions and Outline}

Applying this approach to decoding linear codes poses several
challenges:
\begin{enumerate}
\item[$(i)$] Weitz's construction is valid only
for Markov random fields (MRFs) with pairwise interactions and binary
variables. The decoding problem does not fit  this framework.
\item[$(ii)$] The original justification for truncating
the self-avoiding walk tree followed the so-called `strong spatial
mixing' or `uniqueness' condition. This amounts to saying that the
conditional marginal at the root given the variables at depth $t$,
depends weakly on the values of the latter. This is (most of the
times) false in decoding. For a `good' code, the value of bit $i$
in a codeword is completely determined by the values of bits
outside a finite neighborhood around $i$.
\item[$(iii)$] Even worse, we found in numerical simulations that
the original truncation procedure performs poorly in decoding.
\end{enumerate}
The self-avoiding walk tree construction has already motivated
several applications and generalizations in the past few months.
Jung and Shah \cite{JungShah} discussed its relation with BP,
and proposed a distributed implementation.
Mossel and Sly \cite{MosselSly}
used it to estimate mixing times of Monte Carlo Markov
chain algorithms. Finally, and most relevant to the problems listed above,
Nair and Tetali \cite{NairTetali} proposed a generalization to
non-binary variables and multi-variable interactions.
While this generalizations does in principle apply to
decoding, its complexity grows polynomially in the tree size.
This makes it somewhat unpractical in the present context.

In this paper we report progress on the three points above.
Specifically, in Section \ref{sec:DecodingTree} we use
duality to rephrase decoding in terms of a
generalized binary Markov random field. We then show how to generalize
the self-avoiding walk tree construction to this context.
In Section \ref{sec:Truncation} we discuss the problems
arising from the original truncation procedure, and describe
two procedure that show better performances.
Numerical simulations are presented in Section \ref{sec:Simulations}.
 Finally, one of the most interesting perspectives
is to use TP as a tool for analyzing BP and, in particular, comparing it
with MAP decoding. Some preliminary results in this direction
are discussed in Section \ref{sec:Theoretical}.

We should stress that a good part of our simulations concerns the
binary erasure channel (BEC). From a practical point of view, TP decoding is
not an appealing  algorithm in this case. In fact, MAP decoding
can be implemented in polynomial time through, for instance,
Gaussian elimination. The erasure channel is nevertheless a good
starting point for several reasons. $(i)$ Comparison with MAP
decoding is accessible. $(ii)$ We can find a particularly simple
truncation scheme in the erasure case. $(iii)$ Some subtle numerical
issues that exist for general channels disappear for the BEC.

\section{Decoding through the\\
 Self-Avoiding Walk Tree}
\label{sec:DecodingTree}

Throughout this paper we consider binary linear codes of blocklength
$n$ used over a binary-input memoryless channel.  Let
BM$(\epsilon)$, where $\epsilon$ is a noise parameter, denote a generic channel.
 Assume that  $\cY$ is the output alphabet and  let $\{Q(y|x): x\in\{\0t,\1t\}, y\in\cY\}$ denote
  its transition probability.

With a slight
abuse of terminology we shall 
 identify a code
with a particular parity-check matrix $\H$ that represents it,
\begin{eqnarray*}
\Code = \{\ux\in\{\0t,\1t\}^n:\H\ux = \u0t \mod 2\}\, .
\end{eqnarray*}

Therefore, the code is further identified with a Tanner graph $G =
(V,F,E)$ whose adjacency matrix is the parity-check matrix $\H$.
We will denote by $\da \defas \{i\in V:\, (i,a)\in E\}$ the
neighborhood of function (check) node $a$, and write $\da =
(i_1(a),\dots,i_{k(a)}(a))$. Analogously, $\di\defas \{i\in V:\,
(i,a)\in E\}$ indicates the neighborhood of the variable node $i$.
The conditional distribution for the channel output $\uy$ given
the input $\ux$ factorizes according to the graph $G$ (also called
factor graph). It follows immediately from Bayes rule that
$\prob\{\uX = \ux|\uY=\uy\}= \mu^y(\ux)$, where
\begin{eqnarray}
\mu^y(\ux) = \frac{1}{Z(\uy)}\prod_{i\in V} Q(y_i|x_i)
\prod_{a\in F}\ind({\scriptstyle\sum_{j=1}^{k(a)}x_{i_j(a)}=\0t \mod 2}).
\nonumber\hspace{-1.cm}\\
\label{eq:APDistr}
\end{eqnarray}
%
We denote by $\mu^y_i(x_i) = \prob\{X_i=x_i|\uY=\uy\}$ the marginal
distribution at bit $i$. \emph{Symbol MAP decoding} amounts to the following
prescription,
\begin{eqnarray*}
\xh_i^{\sMAP}(\uy) = \arg\max_{x_i\in \{ \0t,\1t\}} \mu^y_i(x_i)\, .
\end{eqnarray*}
Both BP and TP decoders have the same structure, whereby
the marginal  $\mu^y_i(\,\cdot\,)$ is replaced by its approximation,
respectively $\nu^{\sBP}_i(\,\cdot\,)$ or $\nu^{\sTP}_i(\,\cdot\,)$.
%
%
\subsection{Duality and Generalized
Markov Random Field}

We call a \emph{generalized Markov Random Field} (gMRF) over the
finite alphabet $\cX$  a couple $(\cG,\psi)$, where $\cG =
(\cV,\cE)$ is an ordinary graph over vertex set $\cV$, and edge
set $\cE$. Further $\psi = \{\psi_{ij}:\, (i,j)\in \cE;\;\;
\psi_i:\, i\in\cV\}$ is a set of weights indexed by edges and
vertices in $\cG$, $\psi_{ij}:\cX\times\cX\to\reals$,
$\psi_i:\cX\to \reals$. Notice that, unlike for ordinary MRFs, the
edge weights in generalized MRFs are \underline{not} required to be
non-negative.

Given a subset $A\subseteq\cV$, the \emph{marginal} of the gMRF
$(\cG,\omega)$ on $A$, is defined as the function
$\omega_A:\cX^A\to \reals$, with entries
\begin{eqnarray}
\omega_A(\ux_A) = \sum_{\{x_j:j\not \in A\}}
\prod_{(l,k)\in\cE}\psi_{lk}(x_l,x_k)
\prod_{l\in\cV}\psi_{l}(x_l)\, .\label{eq:MargDef}
\end{eqnarray}
When $A=\cV$, we shall omit the subscript and call $\omega(\ux)$
the \emph{weight} of configuration $\ux$.
More  generally, the \emph{expectation} of a function
$f:\cX^{\cV}\to\reals$ can be defined as
\begin{eqnarray}
\omega(f) & = & \sum_{\ux} f(\ux) \prod_{(l,k)\in\cE}\psi_{lk}(x_l,x_k)
\prod_{l\in\cV}\psi_{l}(x_l)\, .
\end{eqnarray}
Notice that $\omega(\,\cdot\,)$ is not (and in general cannot be)
normalized. 
In the sequel, whenever the relevant MRF has
non-negative weights and is normalizable, we shall use words
`expectation' and `marginal' in the usual (normalized) sense.

Duality can be used to reformulate decoding (in particular, the
posterior marginals $\mu^y_i(x_i)$) in terms of a gMRF.
More precisely, given a code with Tanner graph $G = (V, F, E)$,
we define a gMRF on graph $\cG = (\cV,\cE)$ where
$\cV = (V,F)$ and $\cE = E$, proceeding as follows.
We let  $\cX=\{\0t,\1t\}$ and associate variables $x_i\in\{\0t,\1t\}$
to $i\in V$ and $n_a\in\{\0t,\1t\}$ to $a\in F$. We then introduce the weights
\begin{align}
&\forall i\in V, ~\psi_i(x_i) = Q(y_i|x_i), ~~~ \forall a\in F, ~\psi_a(n_a) = 1, \label{eq:gMRFdec1}\\
&\forall (i,a)\in E, ~\psi_{ai}(n_a,x_i) = (-1)^{n_ax_i}.   \label{eq:gMRFdec2}
\end{align}
Although next statement follows from general duality theory,
it is convenient to spell it out explicitly.
\begin{lemma}
The marginals of the a posteriori distribution defined in Eq.~(\ref{eq:APDistr})
are proportional to the ones of the gMRF defined in
Eq.~(\ref{eq:gMRFdec1}) and Eq.~(\ref{eq:gMRFdec2}).
More 
precisely,  we get $\mu_i^y(x_i) =\omega_i(x_i)/[\omega_i(\0t)+\omega_i(\1t)]$.
\end{lemma}
\begin{proof}
It is immediate to prove a stronger result, namely that
the distribution  $\mu^y(\ux)$ is proportional to $\sum_{\un}\omega(\ux,\un)$,
where $\omega(\ux,\un)$ is defined using Eq.~(\ref{eq:gMRFdec1}) and
Eq.~(\ref{eq:gMRFdec2}). We have
\begin{align*}
\sum_{\un}\omega(\ux,\un) & = \sum_{\un}\prod_{i\in V}Q(y_i|x_i)
\prod_{(i,a)\in E}(-1)^{x_in_a}\\
&=\prod_{i\in V}Q(y_i|x_i)\prod_{a\in F}\sum_{n_a\in\{\0t,\1t\}}
(-1)^{n_a\sum_{i\in\da}x_i} \\
&= \prod_{i\in V}Q(y_i|x_i)\prod_{a\in F}
2\,\ind({\scriptstyle \sum_{i\in\da}x_i=\0t \mod 2})\, ,
%
\end{align*}
which is proportional to the right-hand side of
Eq.~(\ref{eq:APDistr}).
\end{proof}
This result, which 
 derives from \cite{HartmannRudolph} and \cite{Battailetal},
motivates us to extend Weitz's construction
to gMRFs.\footnote{More explicitely, we can think of implementing  a binary {\emph{Fourier}} involution as proposed first in  \cite{HartmannRudolph} and \cite{Battailetal} on the graph edges (as later reported  in Eq.~(\ref{eq:RatioEdges})), while processing all graph vertices in a similar way.} This is the object of the next section.
%
%
\subsection{The Self-Avoiding Walk Tree for Generalized
Markov Random Field}

Assume we are given a graph $\cG = (\cV,\cE)$ and a node
$i\in\cV$. We have already described the computation tree  rooted at $i$, which we denote by  $\Tree(i)$.

An `extended self-avoiding walk' (SAW) on a graph $\cG = (\cV,\cE)$,
starting at $i\in\cV$ is a non-reversing walk that
never visits twice the same vertex, except, possibly, for its end-point.
The `self-avoiding walk tree' rooted at $i\in \cV$ is the tree
of all extended self-avoiding walks on $\cG$ starting at $i$. 
It is straightforward to see that $\SAW(i)$ is in fact a finite
sub-tree of $\Tree(i)$. Its size is bounded by $(\Delta-1)^{|\cV|}$,
where $\Delta$ is the maximum degree of $\cG$, and $|\cV|$ the node number.

\begin{figure}[t]
\centering
\setlength{\unitlength}{0.75bp}%
\begin{picture}(200,100)
\put(0,0){\includegraphics[scale=0.75]{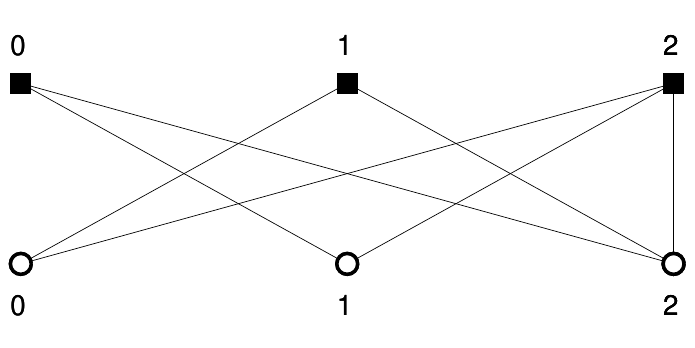}}
\end{picture}
\caption{\small Tanner graph for a repetition code of length $3$.}
\label{fig:Tanner}
\end{figure} %

\begin{figure}[t]
\centering
\setlength{\unitlength}{0.75bp}%
\begin{picture}(210,250)
\put(0,0){\includegraphics[scale=0.75]{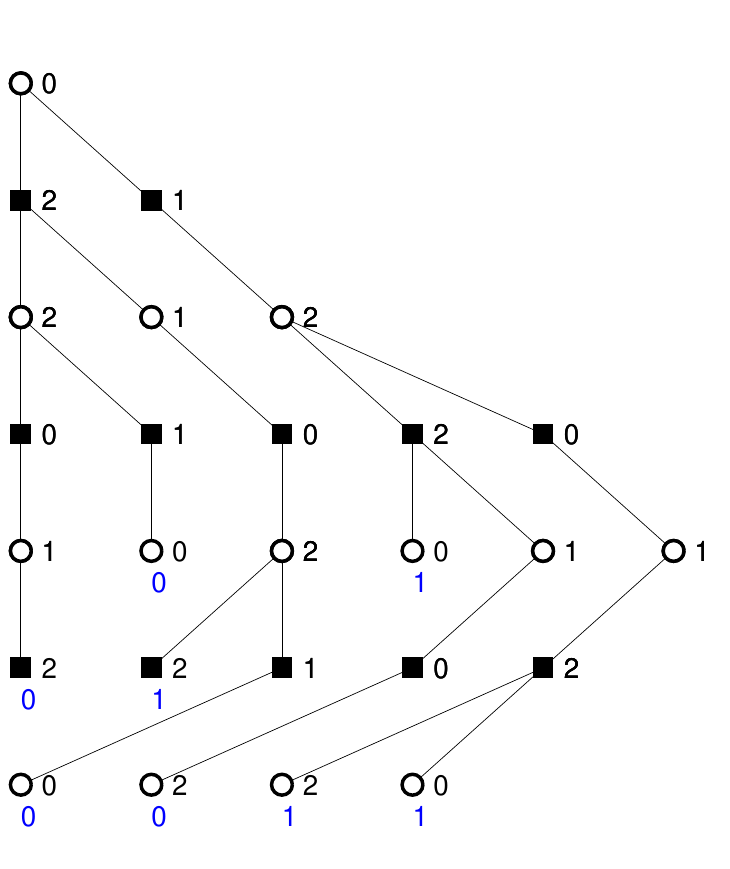}}
\end{picture}
\caption{\small Self-avoiding walk tree  $\SAW(i)$ for the
Tanner graph of Figure \ref{fig:Tanner}, rooted at variable node $i=0$.
In this picture, each node of the self-avoiding tree is labeled by its projection onto $\cV$.
At `terminated' nodes we marked the value that the variable is forced to take.
}
\label{fig:SAWTree}
\end{figure}
As an example,
Figure \ref{fig:SAWTree} shows a SAW tree for the small graph $\cG$
depicted in Figure \ref{fig:Tanner}. (In this case, $\cG$ is the Tanner graph
of a repetition code of length $3$.)
If we denote by $\cV(i)$ the vertex set of $\SAW(i)$, there exists a
natural projection $\pi:\cV(i)\to\cV$ that preserves edges.
Formally, $\pi$ maps a self-avoiding walk to its end-point.

Notice that $\SAW(i)$ has two types of leaf nodes: $(i)$ Nodes
that are leaves in the original graph $\cG$. $(ii)$ Nodes that are
not leaves in $\cG$ but corresponds to extended self-avoiding
walks that cannot be further continued. The latter case arises
when the endpoint of the self-avoiding walk has already been
visited (i.e., when a loop is closed). We shall refer to nodes of
the second type as \emph{terminated} nodes. Indeed, the
self-avoiding walk tree $\SAW(i)$ can be obtained from $\Tree(i)$
by the following \emph{termination} procedure. Imagine descending
$\Tree(i)$ along one of its branches. When the same projection is
encountered for the second time, terminate the branch. Formally,
this means eliminating all the descendants of $u$ whenever $\pi(u)
= \pi(v)$ for some  ancestor $v$ of $u$.

Given a gMRF $(\cG,\psi)$, we can define a gMRF on $\SAW(i)$ in
the usual way. Namely, to any edge $(u,v)\in \SAW(i)$, we
associate a weight coinciding with the one of the corresponding
edge in $\cG$: $\psi_{u,v}(x_u,x_v) =
\psi_{\pi(u),\pi(v)}(x_u,x_v)$. The analogous definition is
repeated for any non-terminated node: $\psi_{u}(x_u) =
\psi_{\pi(u)}(x_u)$. Finally, the choice of weight on terminated
nodes makes use of the hypothesis that $\cX = \{\0t,\1t\}$. Assume
that the edges of $\cG$ are labeled using a given order, e.g., a
lexicographic order. Let $u$ be a terminated node with $\pi(u)=j$.
Then the self-avoiding walk corresponding to $u$ contains a loop
that starts and ends at $j$. We let $\psi_u(x_u) = \ind(x_u=\0t)$
(respectively, $\psi_u(x_u) = \ind(x_u=\1t)$) if this loop quits
$j$ along an edge of higher (respectively, lower) order than the
one along which it enters $j$. The relevance of this construction
is due to Weitz who
considered\footnote{
Weitz \cite{Weitz} considered the \emph{independent set problem}
but remarked that his construction generalized to a larger class
of MRFs. Jung and Shah \cite{JungShah} studied this
generalization. Nair and Tetali \cite{NairTetali} discussed the
case of `hard-core' interactions (positively alignable) as well.}
the case of \emph{permissive} binary MRFs. By this we mean that
$\psi_{kl}(x_k,x_l)\ge 0$, $\psi_k(x_k)\ge 0$, and, for any
$k\in\cV$, there exists $x_k^*\in\{\0t,\1t\}$, such that
$\psi_k(x_k^*)>0,\psi_{kl}(x_k^*,x_l)>0$ for any $l$ with
$(k,l)\in\cE$ and $x_l\in\{\0t,\1t\}$. (The latter is referred to
as the `permissivity' condition.)
\begin{propo}[Weitz]\label{propo:Weitz}
Given a \underline{permissive} binary MRF $(\cG,\psi)$, 
 the marginal of $x_i$ with respect to
$(\cG,\psi)$ is proportional to the root marginal on $\SAW(i)$.
\end{propo}

The problem with non-permissive MRFs and, \emph{a fortiori},
with generalized MRFs, is that the tree model $\SAW(i)$
may not admit any assignment of the variables such that  all the weights
$\psi_l(x_l)$, $\psi_{kl}(x_k,x_l)$ are non-negative. 
As a consequence the MRF on $\SAW(i)$ does not define a probability
distribution and this invalidates the derivation in \cite{Weitz} or \cite{JungShah}.
Even worse, the procedure used in these papers was based in keeping
track of ratios among marginals, of the form $R_i = \mu_i(x_i=\0t)/
\mu_i(x_i=\1t)$. When the MRF does not define a distribution, ill-defined
ratios such as $0/0$ can appear.

Let us stress that this problem is largely due to the
`termination' procedure described above. This in fact constrains
the set of assignments with non-vanishing weight to be compatible
with the values assigned at terminated nodes.

In order to apply the self-avoiding walk construction to gMRFs,
we need to modify it in the two following ways.

$(i)$ We add further structure to $\SAW(i)$.
For any $u\in \cV(i)$, let $\D(u)$ be the set of its
\emph{children} (i.e.,  the set of extended
self-avoiding walks that are obtained by adding one step to $u$).
Then we partition $\D(u) = \D_1(u)\cup\cdots\cup\D_k(u)$ as follows.
Let $v_1,v_2\in \D(u)$ be two children of $u$,
and write them as $v_1 = (u,j_1)$, $v_2 = (u,j_2)$. Further, let
$j=\pi(u)$.
Then we write $v_1\sim v_2$ if there exists an extended
self-avoiding walk of the form $(u,j_1,u',j_2,j)$. Here
we are regarding $u$, $u'$ as walks on $\cG$ (i.e.,  sequences of
vertices) and we use $(u,v,w,\dots)$ to denote the concatenation of walks.
It is not difficult to verify that $\sim$ is an equivalence relation.
The partition $\{\D_1(u),\dots,\D_k(u)\}$ is defined to be the partition in
equivalence classes under this relation.

$(ii)$ We define the \emph{generalized root marginal} of
$\SAW(i)$ through a recursive procedure that makes it always well-defined.
First notice that, if $\cG$ is a tree rooted at $i$, then
the marginal at $i$ can be computed by a standard message passing
(dynamic programming) procedure, starting from the leaves
and moving up to the root. The update rules are, for $u\in\D(w)$,
\begin{eqnarray}
\omega_{u\to w}(x_u) &=& \psi_u(x_u)\prod_{v\in \D(u)}
\tomega_{v\to u}(x_u)\, ,\label{eq:UpNodes}\\
\tomega_{v\to u}(x_u) &=&\sum_{x_{v}}\psi_{uv}(x_u,x_{v})\;\omega_{v\to u}
(x_{v})\, ,\label{eq:UpEdges}
\end{eqnarray}
where edges are understood to be directed towards the root.
The marginal at the root is obtained by evaluating the right
hand side of Eq.~(\ref{eq:UpNodes}) with $u=i$.

The generalized root marginal is defined by the same
procedure but changing Eq.~(\ref{eq:UpNodes}) as follows.
Given the partition $\D(u) = \D_1(u)\cup\cdots\cup\D_k(u)$
described above, we let
\begin{eqnarray}
\omega_{u\to w}(x_u) &=& \psi_u(x_u)\prod_{l=1}^k
\tomega_{\D_l(u)}(x_u)\, ,\label{eq:UpNodes1}
\end{eqnarray}
where we define $\tomega_{\D_l(u)}(x_u)$ through a
\emph{concatenation} procedure. Let
$(\tomega^{(1)}(\, \cdot\,),\dots,\tomega^{(k)}(\,\cdot\,))$
be the set of messages $\{\tomega_{v\to u}(\,\cdot\,):\;
v\in\D(u)\}$ ordered according to the order of edges $(\pi(v),\pi(u))$
in $\cG$. Then we let
\begin{eqnarray}
(\tomega_{\D_l(u)}(\0t),\tomega_{\D_l(u)}(\1t))\defas
(\tomega^{(1)}(\0t),\tomega^{(k)}(\1t))\, .\label{eq:Concatenation}
\end{eqnarray}
The reason for calling this a `concatenation' follows from the remark that,
with the notations above, we have
$\tomega^{(1)}(\1t) = \tomega^{(2)}(\0t)$,
$\tomega^{(2)}(\1t) = \tomega^{(3)}(\0t)$, etc.
We refer to the discussion (and proof) below for a justification of
this claim. As a consequence, the procedure in Eq.~(\ref{eq:Concatenation})
can be described as follows: write the components of $\tomega^{(1)}
(\,\cdot\,),  \tomega^{(2)}(\,\cdot\,)\, \cdots,\tomega^{(k)}(\,\cdot\,)$
in sequence, and eliminate repeated entries.

With this groundwork, we obtain the following generalization of
Weitz's result.
\begin{propo}\label{propo:gMRF}
Given a gMRF $(\cG,\psi)$, the marginal at $i\in\cV$ with respect to
$(\cG,\psi)$ is equal to the generalized root marginal on $\SAW(i)$.
\end{propo}
\begin{proof}
The proof is very similar to Weitz's original proof  in \cite{Weitz}; 
 the 
  difference is that
special care must be paid to avoid ill-defined expressions.
The argument consists in progressively simplifying
the graph $\cG$ (rooted at $i$) until $\SAW(i)$ is obtained.
We shall represent these simplifications graphically.

Consider the first step, corresponding to
Eq.~(\ref{eq:UpNodes1}), with $u=i$.
The partition of $\D(u)$ in $\{\D_1(u)$, $\dots \D_k(u)\}$,
corresponds to a partition of of the subgraph $\cG\setminus i$
(obtained by eliminating from $\cG$, $i$ as well as its adjacent vertices)
into connected components. This correspondence is depicted below
(whereby gray blobs correspond to connected sub-graphs).
\begin{figure}[h!]
\includegraphics[scale=0.225]{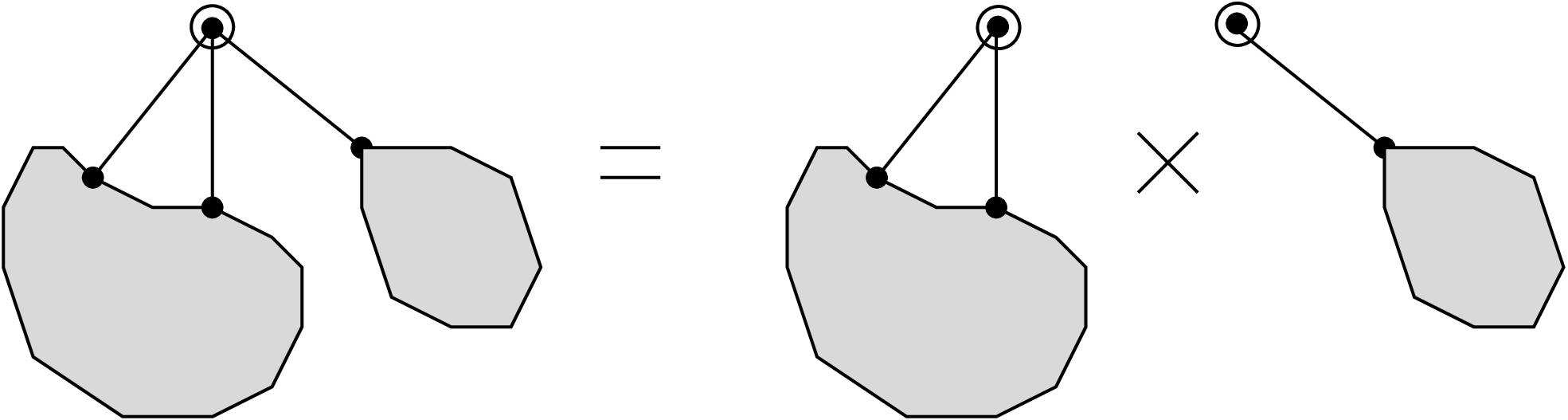}
\put(-195,55){$u$}
\put(-95,55){$u_1$}
\put(-63,55){$u_2$}
\end{figure}
After factoring
out the term $\psi_u(x_u)$, the definition of marginal in
Eq.~(\ref{eq:MargDef}) factorizes naturally on such components,
leading to Eq.~(\ref{eq:UpNodes1})

Consider now one of such components, call it $\cG_1$,
such as the one depicted below.
The corresponding generalized root marginal is computed using
the concatenation rule, specified in Eq.~(\ref{eq:Concatenation}).
\begin{figure}[h!]
\includegraphics[scale=0.225]{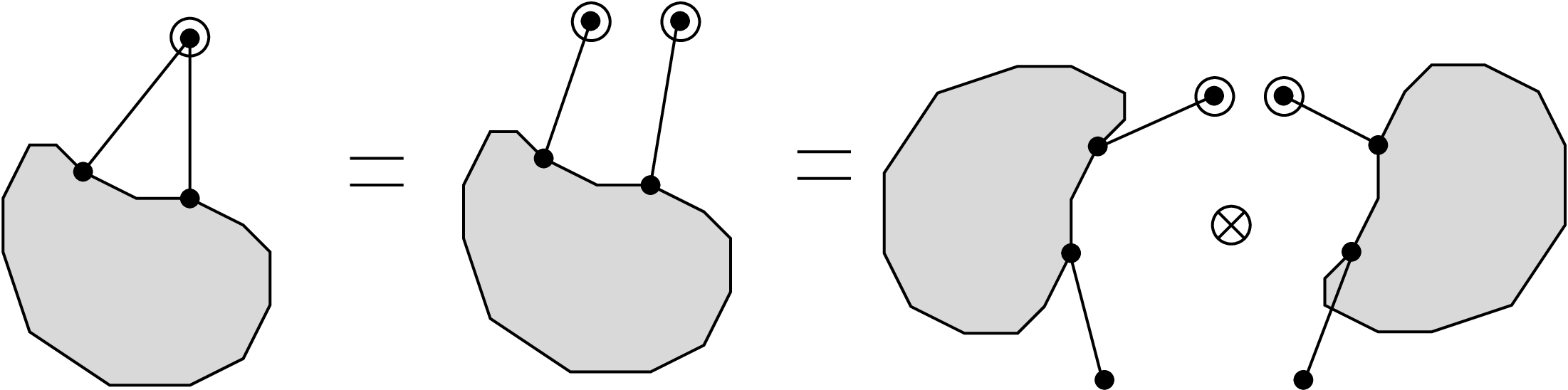}
\put(-225,55){$u$}
\put(-165,55){$u^{(1)}$}
\put(-130,55){$u^{(2)}$}
\put(-65,53){$u^{(1)}$}
\put(-42,53){$u^{(2)}$}
\put(-67,0){$\0t$}
\put(-50,0){$\1t$}
\end{figure}
In order to derive this rule, first consider the graph
$\cG'_1$ obtained from
$\cG_1$ by replacing its root $u$ by $k={\rm deg}(u)$ copies
$u^{(1)},\dots,u^{(k)}$,
each of degree $1$
(here ${\rm deg}(v)$ denotes the degree of vertex $v$).
Each of the newly introduced vertices is adjacent to one of
the edges incident on the root in $\cG_1$.
Further $u^{(1)},\dots,u^{(k)}$ are labeled according to
the ordering (chosen at the beginning of the reduction procedure)
on the adjacent edges. These $k$ nodes will be referred to as
`split nodes' in the sequel.

From the definition of marginal in Eq.~(\ref{eq:MargDef}), and  using the notation $\omega'$
for the gMRF on $\cG_1'$, we have  
\begin{eqnarray}
\omega_u(x) =  \omega'_{u^{(1)}\dots u^{(k)}}(\underbrace{x\cdots\cdots x}_{k})\, ,\label{eq:Splitting}
\end{eqnarray}
for $x\in\{\0t,\1t\}$. This identity is represented as the first
equality in the figure above.

Next we replace the graph $\cG_1'$ by $k$ copies of it,
$\cH_1,\dots,\cH_k$. With a slight abuse of notation,
we re-name $u^{(1)}$ the first of the $k$ `split nodes' in $\cH_1$,
$u^{(2)}$ the second in $\cH_2$, and so on. Further we
add node weights to the other `split nodes,' (i.e.,
the ones that remained un-named), either of the form
$\psi_{v}(x_v)=\ind(x_v=\0t)$ (forcing $x_v$ to take value
$\0t$) or of the form $\psi_v(x_v) = \ind(x_v = \1t)$ (forcing $x_v$
to take value $\1t$). More precisely, for any $j\in\{1,\dots,k\}$ on
$\cH_j$ we force to $\0t$ those split nodes that come before $u^{(j)}$,
and to $\1t$ the ones that come after.

As a consequence, if we use $\omega^{(j)}$  for the gMRF $\cH^{(j)}$,
we have
\begin{eqnarray}
\omega^{(j)}_{u^{(j)}}(x) = \omega'_{u^{(1)}\dots u^{(k)}}(
\underbrace{\0t\cdots\0t}_{j-1}x\underbrace{\1t\cdots\1t}_{k-j})\, .
\end{eqnarray}
In particular, for any $j\in\{1,\dots,k-1\}$,
$\omega^{(j)}_{u^{(j)}}(\1t) =\omega^{(j+1)}_{u^{(j+1)}}(\0t)$.
As a consequence of this fact and of Eq.~(\ref{eq:Splitting}), we get
\begin{eqnarray}
(\omega_u(\0t),\omega_u(\1t)) =
(\omega^{(1)}_{u^{(1)}}(\0t),\omega^{(k)}_{u^{(k)}}(\1t))\, .
\end{eqnarray}
This proves Eq.~(\ref{eq:Concatenation}) with 
$\omega^{(1)}_{u^{(1)}}(x) \defas \tomega^{(1)}(x)$ (second equality in the
last figure above).

Finally, Eq.~(\ref{eq:UpEdges}) follows by considering the marginal of a node
of degree $1$, as $u^{(1)},\dots,u^{(k)}$ in graphs $\cH_1$,
$\dots,\cH_k$, and expressing it in terms of the marginal of its
only neighbor.

This completes one full step of the procedure that breaks the
loops through node $i$. By recursively repeating the same steps,
the graph is completely unfolded giving rise to $\SAW(i)$.
\end{proof}

The self-avoiding walk tree $\SAW(i)$ appears as a convenient
way to organize the calculation of the marginal at $i$ in the general case.
In the case of permissive MRFs this calculation
coincides with a standard marginal calculation on the tree $\SAW(i)$.
It is instructive to check this explicitly.
\begin{fact}
Proposition \ref{propo:Weitz} is a special case of Proposition
\ref{propo:gMRF} for permissive MRFs.
\end{fact}
\begin{proof}
First notice that, for permissive MRFs, the self-avoiding walk tree
construction yields a MRF on $\SAW(i)$ that defines a probability
distribution (non-negative and normalizable),
whose marginals will be denoted as $\omega$ as well.
We have to prove that, in this case, the generalized
root marginal is proportional to the ordinary marginal
at the root of $\SAW(i)$. The crucial remark is that,
because of permissivity, the messages are non-negative and,
in particular, $\omega_{u\to v}(x_u^*)>0$ and $\tomega_{u\to v}(x_v^*)>0$.

Assume, without loss of generality, that $x_u^*=\0t$. We define
the likelihood ratios on the $\SAW(i)$ tree $R_{u\to v} =
\omega_{u\to v}(\1t)/\omega_{u\to v}(\0t)$, $\Rh_{u\to v} =
\hat{\omega}_{u\to v}(\1t)/\hat{\omega}_{u\to v}(\0t)$ and $R_i =
\omega_i(\1t)/\omega_i(\1t)$. The ratio $\Rh_{\D_l(u)}$ is defined
analogously in terms of $\omega_{\D_l(u)}(\, \cdot\,)$. Equation
(\ref{eq:UpEdges}) then implies
\begin{eqnarray}
\Rh_{u\to v} = \frac{\psi_{uv}(\0t,\1t)+\psi_{uv}(\1t,\1t)\, R_{u\to v}}
{\psi_{uv}(\0t,\0t)+\psi_{uv}(\1t,\0t)\, R_{u\to v}}\, .\label{eq:RatioEdges}
\end{eqnarray}
Eq.~(\ref{eq:UpNodes1}) yields on the other hand
\begin{eqnarray}
R_{u\to w} = \frac{\psi_u(\1t)}{\psi_u(\0t)}\prod_{l=1}^k
\Rh_{\D_l(u)}\, .
\end{eqnarray}
Finally, using the remark that $\tomega^{(l)}(\1t) = \tomega^{(l+1)}(\0t)$
for $l=1,\dots,k-1$, we get from Eq.~(\ref{eq:Concatenation})
\begin{eqnarray}
\Rh_{D_l(u)} = \Rh^{(1)}\cdots \Rh^{(k)} =
\prod_{v\in\D_l(u)}\Rh_{v\to u}\, .
\end{eqnarray}
Putting the last two equations together
\begin{eqnarray}
R_{u\to w} = \frac{\psi_u(\1t)}{\psi_u(\0t)}\prod_{v\in\D(u)}
\Rh_{v\to u}\, .\label{eq:RatioNodes}
\end{eqnarray}
It is now easy to check that,  Eq.~(\ref{eq:RatioNodes}) and
Eq.~(\ref{eq:RatioEdges}) coincide with the appropriate recursive
definition of probability marginal ratios on $\SAW(i)$.
\end{proof}

Proposition \ref{propo:gMRF} does not yield an efficient
way of computing marginals of gMRF. The conundrum is that
the resulting complexity is linear in the size of $\SAW(i)$ which is in turn
exponential in the size of the original graph $\cG$.
On the other hand, it provides a systematic way to define and study algorithms
for computing efficiently such a marginal.
The idea, proposed first in \cite{Weitz}, is to
deform $\SAW(i)$ in such a way that its generalized root marginal does
not change too much, but computing it is much easier.
%
%
\section{Truncating the Tree}
\label{sec:Truncation}

BP can be seen as an example of the approach mentioned at the end of
the previous section. In this case $\SAW(i)$ is
replaced by the first $t$ generations of the computation tree, to
be denoted by $\Tree(i;t)$. In this case the complexity of evaluating the
generalized root marginal scales as $t$ rather than as $|\Tree(i;t)|$.

A different idea is to cut some of the branches of
$\SAW(i)$ in such a way to reduce drastically its size.
We will call \emph{truncation} the procedure of cutting branches of
$\SAW(i)$. It is important to keep in mind that truncation is different from
the \emph{termination} of branches when a loop is closed in $\cG$.
While termination is completely defined, we are free to define truncation
to get as good an algorithm as we want.
In the following we shall define truncation schemes parametrized
by an integer $t$, and denoted as $\SCHEME(i;t)$. We will have
$\SCHEME(i;t) = \SAW(i)$ for $t\ge n$, thus recovering the exact marginal by
Proposition \ref{propo:gMRF}.

In order for the algorithm to be 
efficient, we need
to ensure the following constraints. $(i)$ $\SCHEME(i;t)$ is `small enough' (as the complexity of
computing its generalized root marginal is at most linear in its
size). $(ii)$ $\SCHEME(i;t)$ is `easy to construct.'
For coding applications, this second constraint is somewhat less
restrictive because the tree(s) $\SCHEME(i;t)$ can be constructed
in a preprocessing stage and not  recomputed at each use of the code.

In order to achieve the second goal, we must define the partition
$\D(u) = \D_{1,t}(u)\cup\cdots\cup \D_{k,t}(u)$ of children of $u$
according to the subtree $\SCHEME(i;t)$ used in the computation.
Consider two children of $u$,  which we denote by $v_1,v_2\in \D(u)$. 
In a similar way as for the $\SCHEME(i)$ in the complete tree case, we write them as
$v_1 = (u,j_1)$, $v_2 = (u,j_2)$, and define $v_1\sim_t v_2$ if
there exists a descendant $v'_1$ of $v_1$ in $\SCHEME(i;t)$ such
that $\pi(v'_1) = j_2$ or a descendant $v'_2$ of $v_2$ such that
$\pi(v'_2) = j_1$. The construction of the partition
$\{\D_{1,t}(u),\dots,\D_{k,t}(u)\}$ will be different whether
communication takes place over erasure or general channels.
%
%
\subsection{Weitz's Fixed Depth Scheme and its Problems}

The truncation procedure proposed in \cite{Weitz} amounts to
truncating all branches of $\SAW(i)$ at the same depth $t$,
unless they are already terminated at smaller depth.
Variables at depth $t$ (boundary) are forced to take arbitrary values.

The rationale for this scheme comes from the `strong spatial mixing'
property that holds for the system studied in \cite{Weitz}.
Namely, if we denote by $\omega_{i|t}(x_i|\ux_t)$ the normalized
marginal distribution at the root $i$ given variable assignment at depth
$t$, we have
\begin{eqnarray}
||\omega_{i|t}(\,\cdot\,|\ux_t)-\omega_{i|t}(\,\cdot\,
|\ux_t')||_{\sTV}\le A\, \lambda^t\, ,\label{eq:SSM}
\end{eqnarray}
uniformly in the boundary conditions $\ux_t$, $\ux'_t$
for some constants $A>0$, $\lambda\in [0,1)$.

It is easy to realize that the condition in Eq.~(\ref{eq:SSM})
generically does not hold for `good' sparse graph codes.
The reason is that fixing the codeword values at the boundary
of a large tree, normally determines their values inside the same tree.
In other words the TP estimates strongly depend on this boundary condition.

One can still hope that some simple boundary condition might
yield empirically good estimates. An appealing choice is to
leave `free' the nodes at which the tree is truncated.
This means that no node potential is added on these boundary vertices.
We performed numerical simulations with this scheme on the same examples
considered in the next section. The results are rather poor:
unless the truncation level $t$ is very large (which is feasible only
for small codes in practice) the bit error rate is typically worse than under BP decoding.
%
%
\subsection{Improved Truncation Schemes: Erasure Channel}
\label{sec:ErasureTruncation}

For decoding over the BEC, a simple trick improves
remarkably the performances of TP decoding. First, construct a
subtree of $\SAW(i)$ of depth at most $t$ by truncating at the
deepest variable nodes whose depth does not exceed $t$. The
partition $\{\D_{1,t}(u),\dots,\D_{k,t}(u)\}$ is constructed using
the equivalence class of the transitive closure of $\sim_t$. Then
run ordinary BP on this graph, upwards from the leaves towards the
root, and determine all messages in this direction. If a node $u$
is decoded in this way, fix it to the corresponding value and
further truncate the tree $\SAW(i)$ at this node.

The resulting tree $\SAW(i;t)$ is not larger than the one
resulting from fixed-depth truncation. For low erasure probabilities
it is in fact much smaller than the latter.
%
%
\subsection{Improved Truncation Schemes: General channel}
\label{sec:TruncationGeneral}

The above trick cannot be applied to general BM
channels. We therefore resort to the following two constructions.

$(i)$ {Construction $\MAPt(i;t)$}: Define the distance $d(i,j)$ between two variable
nodes $i$, $j$ to be the number of check nodes encountered along
the shortest path from $i$ to $j$. Let  $\ball(i;t)$ be the
subgraph induced\footnote{The subgraph induced by a subset $U$ of
variable nodes is the one including all those check nodes that
only involve variables in $U$.} by  variable nodes whose distance
from $i$ is at most $t$. Then we let $\SAW(i;t)\defas \MAPt(i;t)$
be the \emph{complete} self-avoiding walk tree for the subgraph
$\ball(i;t)$.
This corresponds to truncating $\SAW(i)$ as soon as the
corresponding self-avoiding walk exits $\ball(i;t)$. No forcing
self-potential is added on the boundary.

A nice property of this scheme is that it returns the \emph{a
posteriori} estimate of transmitted bit $X_i$ given the channel
outputs within $\ball(i;t)$, call it $\underline{Y}_{\ball(i;t)}$.
As a consequence, many reasonable performance measures (bit error
probability, conditional entropy, etc.) are monotone in $t$
\cite{TwoLectures}.

On the negative side, the size of the tree $\MAPt(i;t)$ grows very
rapidly (doubly exponentially) with $t$ at small $t$. This
prevented us from using $t\ge 3$. 


$(ii)$ {Construction $\MAPt(i;t)-\BPt(\ell)$}: The tree $\MAPt(i;t)$ constructed as in
the previous approach is augmented by adding some descendants to
those nodes that are terminated in $\MAPt(i;t)$. More precisely,
below any such node $u$ in the mentioned tree, we add the first
$\ell$ generations of the computation tree. 

$(iii)$ Construction $\SAW(i;t)$: 
We can implement a finer scheme for the general BM case. 
This scheme operates on $\SAW(i;t)$ obtained by
truncating all branches of $\SAW(i)$ at the same depth $t$. The description of 
this method is slightly lengthy. We
omit the details here and choose to present the numerical results in Fig.~\ref{fig:exitcurvesconvc50bawgn} of Section
\ref{sec:Simulations}.
%
%
\section{Numerical simulations}
\label{sec:Simulations}

For communication over the BEC, the implementation
of TP decoding as described in Section \ref{sec:ErasureTruncation} is satisfying.
While it is simple enough for practical purpose, it permits us to depict performance
curves that interpolate successfully between BP and MAP decoding.
 The binary erasure channel, which we denote by BEC$(\e)$ if the erasure probability is $\e$, is appealing for a first study for the following reasons. $(i)$
 The accessibility of performance curves under MAP
decoding allows for a careful study of the new algorithm. $(ii)$
The TP decoder turns  out to be `robust' with respect to
changes in the truncation method, hence simpler to study.

As an example for a generic BM channel, we shall
consider the binary-input additive white Gaussian noise channel
with standard deviation $\sigma$, which we denote by
BAWGN($\sigma^2$).

Let us stress that the TP decoder is not (in general) symmetric
with respect to codewords. This complicates a little the analysis
(and simulations) which has to be performed for uniformly random
transmitted codewords.
%
%
\subsection{Binary Erasure Channel}

The erasure case is illustrated
by three examples: a tail-biting convolutional code,
the $(23,12)$ Golay code and a sparse graph code.
Here the comparison is done with BP after convergence
(`infinite' number of iterations) and MAP as implemented through
Gaussian elimination.

TP decoder permits us to plot a sequence of performance curves,
indexed by the truncation parameter $t$. In all of the cases
considered, TP improves over BP already at small values of $t$. As
$t$ increases, TP eventually comes very close to MAP. The gain is
particularly clear in codes with many short loops, and at low
noise. This confirms the expectation that, when truncated, TP
effectively takes care of small `pseudocodewords.'

The first example is a memory two and rate $1/2$ convolutional
code in tailbiting form with blocklength $\n=100$. The
performance curves are shown in Fig.~\ref{fig:ccm2_100}.
The TP and BP decoders are based on a periodic Tanner graph
associated with the tailbiting code with generator pair
$(1+D^2,1+D+D^2)$. More precisely, they are based on the graph
representing the parity-check matrix with circulant
horizontal pattern  $\1t\1t\0t\1t\1t\1t$.

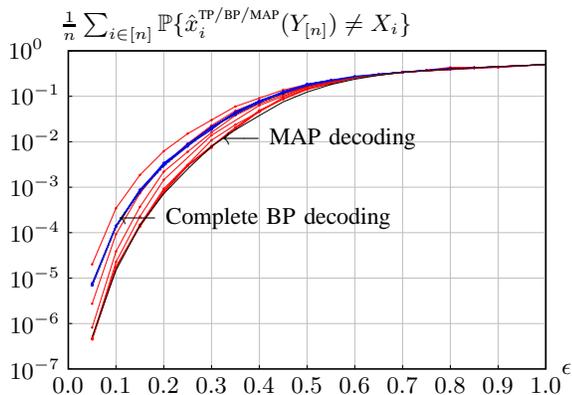
\begin{figure}[hbt]
\centering
\setlength{\unitlength}{0.75bp}%
\begin{picture}(260,200)
\put(0,0)
{
\put(10,10){\includegraphics[scale=0.75]{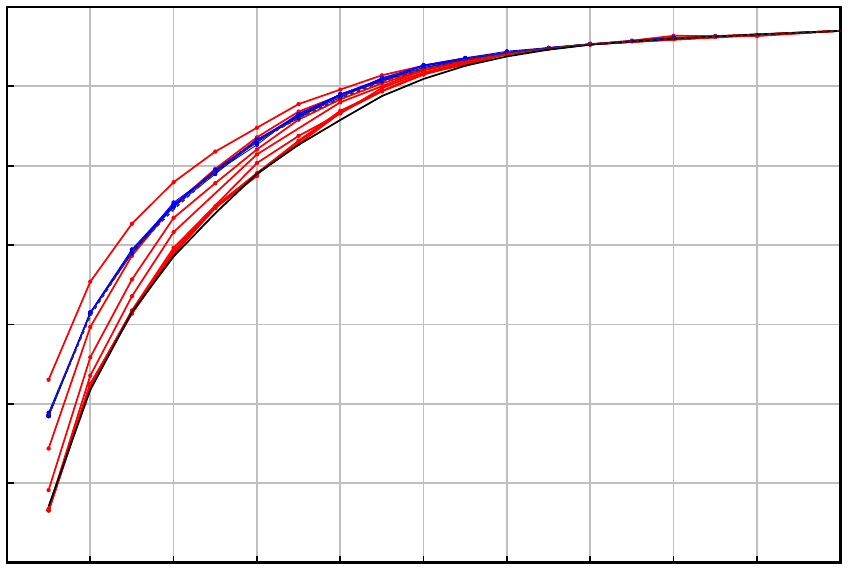}}

\small
\multiputlist(12,0)(24,0)[cb]{$0.0$,$0.1$,$0.2$,$0.3$,$0.4$,$0.5$,$0.6$,$0.7$,$0.8$,$0.9$,$1.0$}
\multiputlist(8,12)(0,23)[rc]{$10^{-7}$,$10^{-6}$,$10^{-5}$,$10^{-4}$,$10^{-3}$,$10^{-2}$,$10^{-1}$,$10^{0~}~$}
\put(260,12){\makebox(0,0)[l]{\e}} 
\put(8,185){\makebox(0,0)[l]{$\frac1\n \sum_{i\in[n]}\prob\{\hat{x}^{\TP/\BP/\MAP}_i(Y_{[n]})\neq X_i\}$}}

\put(89,125){\rotatebox{0.0}{\small $\longleftarrow$ MAP decoding}}
\put(37,85){\rotatebox{0.0}{\small $\longleftarrow$ Complete BP decoding}}
}
\end{picture}
\caption{\small Tailbiting convolutional code with generator pair
$(1+D^2,1+D+D^2)$ and blocklength $\n=100$.  Black curve:
BP decoding with $t=\infty$. Red curve: MAP decoding (BP and Gaussian elimination). Blue curves: BP decoding with $t=3,4,5,6,8,10,12,14$ (almost indistinguishable). Red curves: TP decoding  with $t=3,4,5,6,8,10,12,14$ (truncated tree). }
\label{fig:ccm2_100}
\end{figure} %

The second example is the standard (perfect) Golay code with blocklength $n=23$. It is shown in Fig.~\ref{fig:cgolay_23}.

\begin{figure}[hbt]
\centering
\setlength{\unitlength}{0.75bp}%
\begin{picture}(260,200)
\put(0,0)
{
\put(10,10){\includegraphics[scale=0.75]{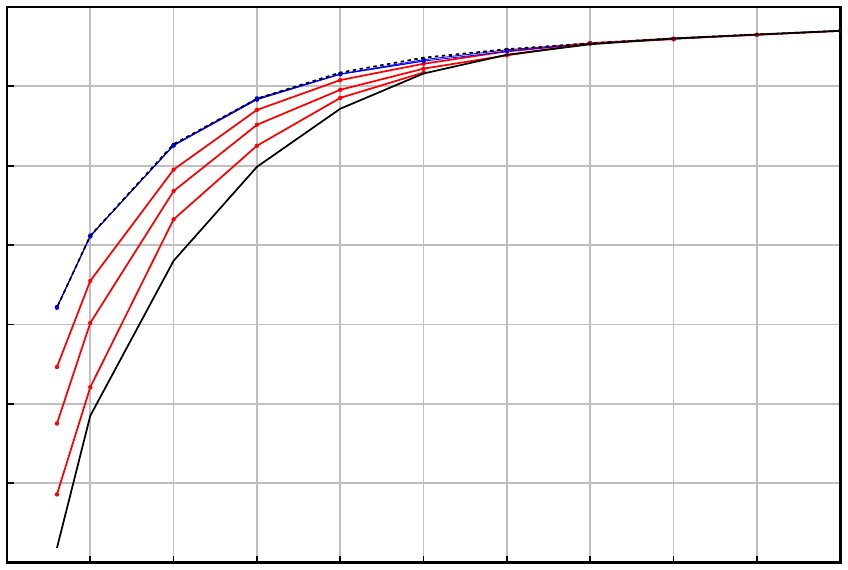}}

\small
\multiputlist(12,0)(24,0)[cb]{$0.0$,$0.1$,$0.2$,$0.3$,$0.4$,$0.5$,$0.6$,$0.7$,$0.8$,$0.9$,$1.0$}
\multiputlist(8,12)(0,23)[rc]{$10^{-7}$,$10^{-6}$,$10^{-5}$,$10^{-4}$,$10^{-3}$,$10^{-2}$,$10^{-1}$,$10^{0~}~$}
\put(260,12){\makebox(0,0)[l]{\e}} 
\put(8,185){\makebox(0,0)[l]{$\frac1\n \sum_{i\in[n]}\prob\{\hat{x}^{\TP/\BP/\MAP}_i(Y_{[n]})\neq X_i\}$}}

}
\end{picture}
\caption{\small $(23,12)$ Golay code with blocklength $\n=23$.
Blue curve: BP decoding with $t=\infty$.
Black curve: MAP decoding (BP and Gaussian
elimination). Blue curves: BP decoding with $t=4,5,6$.
Red curves: TP decoding  with $t=4,5,6$ (truncated tree). }
\label{fig:cgolay_23}
\end{figure}
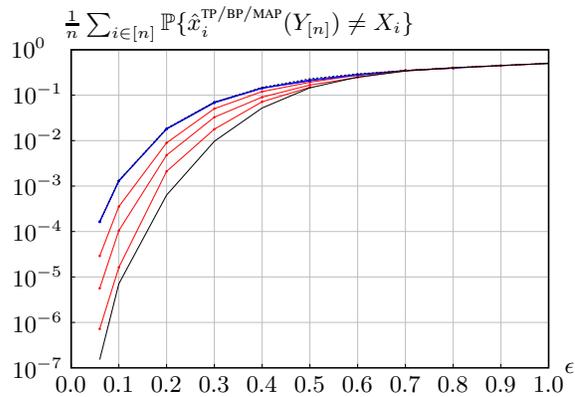 %

The third example, an LDPC code with blocklength $n=50$, is depicted in Fig.~\ref{fig:creg36_50}.

\begin{figure}[hbt]
\centering
\setlength{\unitlength}{0.75bp}%
\begin{picture}(260,200)
\put(0,0)
{
\put(10,10){\includegraphics[scale=0.75]{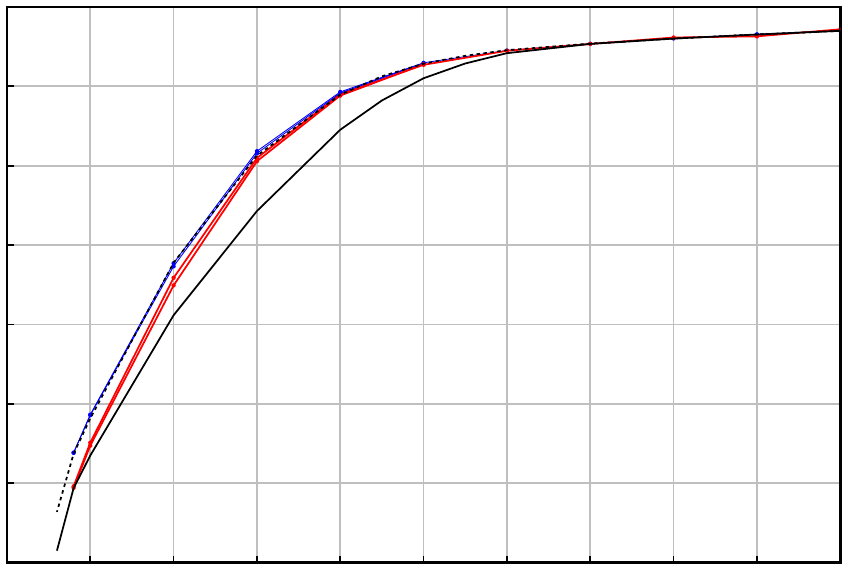}}

\small
\multiputlist(12,0)(24,0)[cb]{$0.0$,$0.1$,$0.2$,$0.3$,$0.4$,$0.5$,$0.6$,$0.7$,$0.8$,$0.9$,$1.0$}
\multiputlist(8,12)(0,23)[rc]{$10^{-7}$,$10^{-6}$,$10^{-5}$,$10^{-4}$,$10^{-3}$,$10^{-2}$,$10^{-1}$,$10^{0~}~$}
\put(260,12){\makebox(0,0)[l]{\e}} 
\put(8,185){\makebox(0,0)[l]{$\frac1\n \sum_{i\in[n]}\prob\{\hat{x}^{\TP/\BP/\MAP}_i(Y_{[n]})\neq X_i\}$}}
}
\end{picture}
\caption{\small A regular $(3,6)$ LDPC code with blocklength $\n=50$.  Blue curve: BP decoding with $t=\infty$. Black curve: MAP decoding (BP and Gaussian elimination). Blue curves: BP decoding with $t=7,8$. Red curves: TP decoding  with $t=7,8$ (truncated tree). }
\label{fig:creg36_50}
\end{figure}
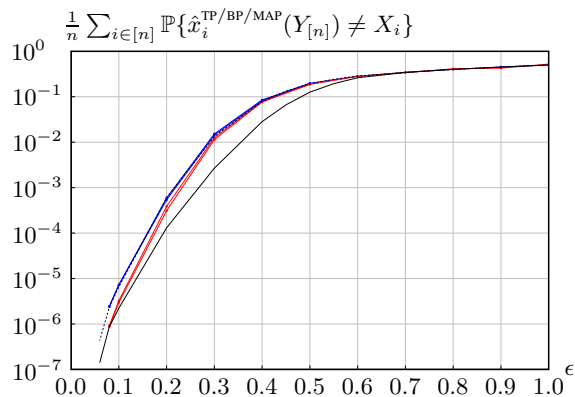
%
%
\subsection{Binary-Input Additive White Gaussian Noise Channel}

In the case of the BAWGN channel, we consider a single example of code, 
the tail-biting convolutional code used above, and two truncation schemes, 
 the constructions $(ii)$ and $(iii)$ described in Section \ref{sec:TruncationGeneral}.

Our results are shown in Fig.~\ref{fig:ccm2_50_bawgn} and
Fig.~\ref{fig:exitcurvesconvc50bawgn}. The TP and BP decoders are
based on the natural periodic Tanner graph associated with the
tailbiting code. We run BP a large number of iterations and check
the error probability to be roughly independent of the iterations
number. The MAP decoder is performed using BP on the single-cycle
tailbiting trellis (i.e.,  BCJR on a ring
\cite{McEliece,HartmannRudolph,AndersonHladik}).

 We observe that the two schemes $\MAPt(i;t)-\BPt(\ell)$ and $\SAW(i;t)$ with $t=8$ outperform
BP. Unhappily, due to complexity constraints we were limited to
small values of $t$ and therefore could not approach the actual
MAP performances.

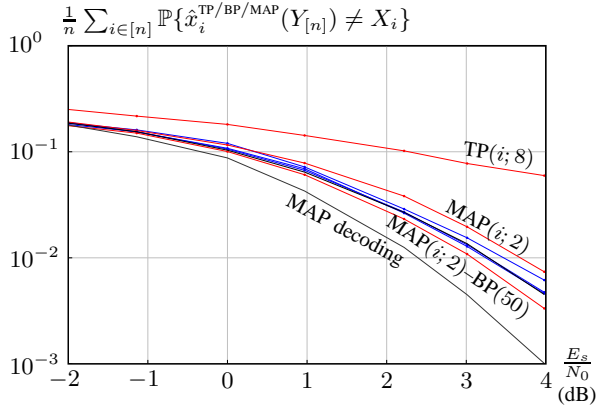
\begin{figure}[hbt]
\centering
\setlength{\unitlength}{0.75bp}%
\begin{picture}(260,200)
\put(0,0)
{
\put(10,10){\includegraphics[scale=0.75]{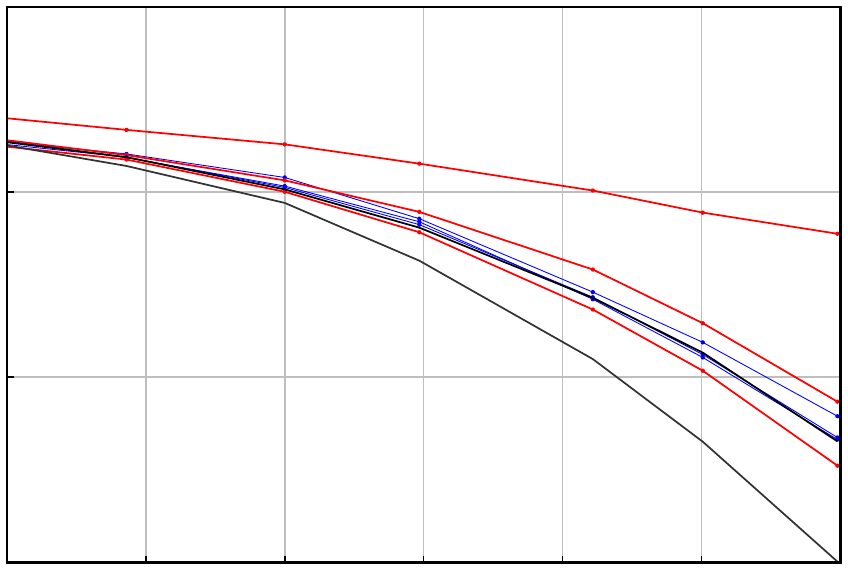}}

\small
\multiputlist(10,0)(40,0)[cb]{$-2$,$-1$,$~0$,$~1$,$~2$,$~3$,$~4$}
\multiputlist(8,12)(0,54)[rc]{$10^{-3}$,$10^{-2}$,$10^{-1}$,$10^{0~}~$}
\put(260,12){\makebox(0,0)[l]{$\frac{E_s}{N_0}$}}
\put(258,-4){\makebox(0,0)[l]{\footnotesize(dB)}}  
\put(8,185){\makebox(0,0)[l]{$\frac1\n \sum_{i\in[n]}\prob\{\hat{x}^{\TP/\BP/\MAP}_i(Y_{[n]})\neq X_i\}$}}
\put(120,92){\rotatebox{-28}{\footnotesize  MAP decoding}}
\put(210,117){\rotatebox{-8}{\footnotesize TP$(i;8)$}}
\put(200,89){\rotatebox{-28}{\footnotesize MAP$(i;2)$}}
\put(170,82){\rotatebox{-34}{\footnotesize MAP$(i;2)$--BP${(50)}$}}
}
\end{picture}
\caption{\small Tailbiting convolutional code with generator pair $(1+D^2,1+D+D^2)$ and blocklength $\n=50$. Dashed black curve: BP decoding with $t=400$. Black curve: MAP decoding (wrap-around BCJR). Blue curves: BP decoding with $t=8,50$. Red curves: TP decoding  with $t=8$ (truncated tree, denoted by TP($i;8)$), TP decoding on a ball of radius 2 (scheme $(ii)$ with no BP processing, denoted by  $\MAPt(2)$) and, TP decoding according to scheme $(ii)$ (with parameters as indicated by  $\MAPt(i;2)-\BPt(50)$).
 }
\label{fig:ccm2_50_bawgn}
\end{figure} %

\begin{figure}[hbt]
\centering
\setlength{\unitlength}{0.75bp}
\begin{picture}(260,200)
\put(0,0)
{
\put(10,10){\includegraphics[scale=0.75]{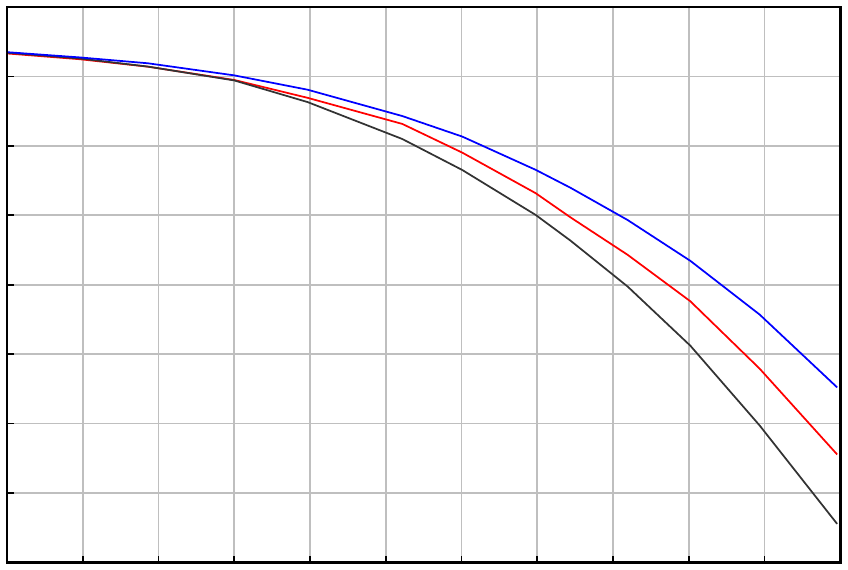}}
\small
\multiputlist(11,0)(22,0)[cb]{$-3$,$-2$,$-1$,$0$,$1$,$2$,$3$,$4$,$5$,$6$,$7$,$8$}
\multiputlist(8,12)(0,20)[rc]{$10^{-8}$,$10^{-7}$,$10^{-6}$,$10^{-5}$,$10^{-4}$,$10^{-3}$,$10^{-2}$,$10^{-1}$,$10^{0~}~$}
\put(260,12){\makebox(0,0)[l]{$\frac{E_s}{N_0}$}}
\put(258,-4){\makebox(0,0)[l]{\footnotesize(dB)}}
\put(8,185){\makebox(0,0)[l]{$\frac1\n \sum_{i\in[n]}\prob\{\hat{x}^{{\TP}/{\BP}/\MAP}_i(Y_{[n]})\neq X_i\}$}}
\put(190,114){\rotatebox{-33}{\footnotesize BP decoding}}
\put(180,104){\rotatebox{-38}{\footnotesize $\SCHEME(i;8)$}}
\put(170,94){\rotatebox{-43}{\footnotesize MAP decoding}}
}
\end{picture}
\caption{\small Tailbiting convolutional code with generator pair
$(1+D^2,1+D+D^2)$ and blocklength $\n=50$. Blue curve: BP
decoding with $t=400$. Black curve: MAP decoding  (wrap-around BCJR). Red
curve: TP decoding according to scheme $(iii)$ (with parameter as indicated by  $\SCHEME(i;t=8)$ using a suitable truncated tree).
 }
\label{fig:exitcurvesconvc50bawgn}
\end{figure}
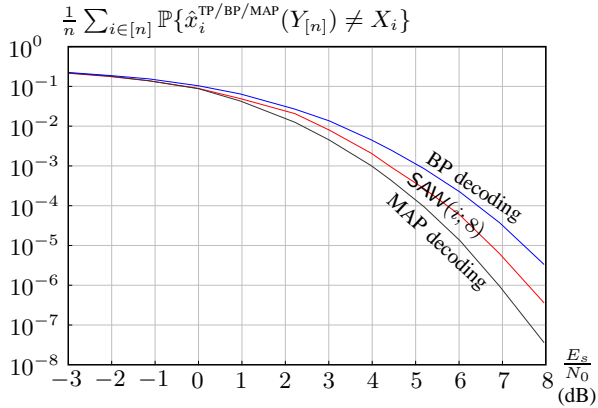

%
%
\section{Theoretical Implications}
\label{sec:Theoretical}


One interesting direction is to use 
 the self-avoiding walk
tree construction for analysis purposes. We think in particular
of two types of developments: $(i)$ a better understanding of the relation between
BP and MAP decoding, and $(ii)$ a study of the `inherent hardness'
of decoding sparse graph codes.

While the first point is self-explanatory, it might be useful to
spend a few words on the second. The most important outcome
of the theory of iterative coding systems can be phrased as follows.
\begin{quote}
There exist families of graphs (expanders \cite{SipserSpielman}, random
\cite{Capacity})
with diverging size and bounded degree,
below a certain noise level, MAP decoding can be achieved in linear time
up to a `small error.'
\end{quote}
We think (a formal version of) the same statement to be true for
\emph{any family of graphs} with bounded degree. 
This can be proved for
the erasure channel.
\begin{propo}
Let $\{G_n\}$ be a family of Tanner graphs of diverging
blocklength $n$, with maximum variable degree ${\tt l}$ and
check degree ${\tt r}$. Consider communication over
{\rm BEC}$(\epsilon)$ with $\epsilon<1/(\tl-1)(\tr-1)$.
Then, for any $\delta>0$ there exists a decoder
whose complexity is of order $n{\rm Poly}(1/\delta)$ and
returning estimates $\{\xh_1(\uy),\xh_2(\uy),\dots,\xh_n(\uy)\}$
such that  $\prob\{\xh_i(\uy)\neq\xh^{\MAP}_i(\uy)\}\le \delta$.
\end{propo}
\begin{proof}
The decoder consists in returning the MAP estimate
of $i$ given the subgraph $\ball(i;t)$ and the values received therein.
Consider the subgraph $G_n(\uy)$ of
$G_n$ obtained by removing non-erased bits. The proof
consists in an elementary percolation estimate on this
graph, see \cite{Grimmett}.

It is easy to see that $\prob\{\xh_i(\uy)\neq\xh^{\MAP}_i(\uy)\}$
is upper bounded by the probability that the connected component
of $G_n(\uy)$ that contains $i$ is not-contained in $\ball(i;t)$.
This is in turn upper bounded
by the number of paths between $i$ and a vertex at distance $t+1$
(which is at most $\tl(\tl-1)^t(\tr-1)^t$) times the probability that
one such path is completely erased (which is $\epsilon^{t+1}$).
Therefore, for $A=\tl\epsilon>0$ and $\lambda = (\tl-1)(\tr-1)\epsilon<1$, 
we get 
%
%
$\prob\{\xh_i(\uy)\neq\xh^{\MAP}_i(\uy)\}\le A\lambda^t\, .$ 
%
%
The proof is completed by taking $t = \log(A/\delta)/\log(1/\lambda)$,
and noticing that $\ball(i;t)$ can be decoded in time
polynomial in its size, that is polynomial in $1/\delta$. The
computation is repeated for each $i\in\{1,\dots,n\}$ whence the
factor $n$.
\end{proof}
We think that a strengthening (better dependence on the precision $\delta$)
and generalization (to other channel models) of this result
can be obtained using the self-avoiding walk tree construction.

\section{Acknoledgments} Yi Lu is supported by the Cisco Stanford Graduate Fellowship.
Cyril M{\'e}asson is supported by the Swiss National Fund Postdoctoral Fellowship.
%
%

\addtolength{\textheight}{-3cm}

\end{document}